\documentclass[12pt]{article}
\usepackage{color}
\usepackage[latin1]{inputenc}
\usepackage[T1]{fontenc}
\usepackage[english]{babel}
\usepackage{amsfonts}
\usepackage{amssymb}
\usepackage{graphicx}
\usepackage{dsfont}
\usepackage{epsfig}
\usepackage{ae}
\usepackage{hyperref}
\usepackage{harvard}
\usepackage{amsmath}

\usepackage{amsthm}
\usepackage{appendix}

\renewcommand{\(}{\left(}

\newcommand{\1}{\mathds{1}}

\newtheorem{theorem}{Theorem}[section]
\newtheorem{proposition}{Proposition}[section]
\newtheorem{lem}{Lemma}[section]
\newtheorem{co}{Corollary}[section]
\newtheorem{re}{Remark}[section]

\newcommand{\be}{\begin{eqnarray}}
\newcommand{\ee}{\end{eqnarray}}
\newcommand{\by}{\begin{eqnarray*}}
\newcommand{\ey}{\end{eqnarray*}}
\newcommand{\bt}{\begin{theo}}
\newcommand{\et}{\end{theo}}
\newcommand{\bl}{\begin{lem}}
\newcommand{\el}{\end{lem}}
\newtheorem*{lem1}{Lemma A.1}

\newtheorem*{lem2}{Lemma A.2}

\newtheorem*{lem3}{Lemma A.3}

\newcommand{\bc}{\begin{co}}
\newcommand{\ec}{\end{co}}
\newcommand{\eex}{\end{exa}\vspace{-3mm}}
\newcommand{\br}{\begin{re}}
\newcommand{\er}{\end{re}\vspace{-3mm}}
\renewcommand{\geq}{\geqslant}
\renewcommand{\leq}{\leqslant}

\begin{document}

\bibliographystyle{econometrica}
\citationstyle{agsm}

\title{Omega risk model with tax}
\author{   \textsc{Zhenyu Cui} \footnote{Corresponding Author. Zhenyu
Cui is with the department of Mathematics at Brooklyn College of the City University of New York, Ingersoll Hall, 2900 Bedford Ave, Brooklyn, NY11210, United States.   Tel.: +1718-951-5600, ext. 6892
              Fax: +1718-951-4674.
              Email: \texttt{zhenyucui@brooklyn.cuny.edu} } \ }

\date{Draft: \today }
\maketitle

\begin{abstract}
In this paper we study the Omega risk model with surplus-dependent tax payments in a time-homogeneous diffusion setting. The new model incorporates practical features from both the Omega risk model(Albrecher and Gerber and Shiu \citeyear{AGS11}) and the risk model with tax(Albrecher and Hipp \citeyear{AH07}). We explicitly characterize the Laplace transform of the occupation time of an Azema-Yor process(e.g. a process refracted by functionals of its running maximum) below a constant level until the first hitting time of another Azema-Yor process or until an independent exponential time. This result unifies and extends recent literature(Li and Zhou \citeyear{LZ13} and Zhang \citeyear{Z13}) incorporating some of their results as special cases. We explicitly characterize the Laplace transform of the time of bankruptcy in the Omega risk model with tax and discuss an extension to integral functionals. Finally we present examples using a Brownian motion with drift.
%For the Omega risk model in risk analysis, please refer to Albrecher and Gerber and Shiu \citeyear{AGS11}, Gerber, Shiu and Yang \citeyear{GSY12}).  For the diffusion risk model with a value-dependent tax rate, please refer to Albrecher and Hipp \citeyear{AH07}, .
\end{abstract}
\vspace{0.5cm}

\noindent{{\bf Math Subject Classification} 60G44 \and 91B25 \and 91B70}

\newpage

%\keywords{Time-homogeneous diffusion; Azema-Yor process; occupation time; Laplace transform; risk model with tax; Omega risk model}\vspace{3mm}
%
%\ams{60G44}{91B25; 91B70} % insert the primary Maths Subject Classification number in the first bracket
%         % and the secondary ams number(s) in the second bracket
%         % e.g. \ams{60E20}{49G03;49F10}

\newpage

\section{Introduction}

The Omega risk model was first introduced by Albrecher, Gerber and Shiu \citeyear{AGS11}, and it distinguishes the ruin time(negative surplus) from the time of bankruptcy of a company(occupation time  of the negative surplus exceeds a \textit{grace period}).

The risk model with tax was first introduced in Albrecher and Hipp \citeyear{AH07}, where a constant tax rate is applied to the compound Poisson risk model at profitable times. In a time-homogeneous diffusion setting, Li, Tang and Zhou \citeyear{LTZ13} introduce a diffusion risk model with tax and model the ruin time of the company by its two-sided exit time. In the Levy insurance model with tax, Kyprianou and Zhou \citeyear{KZ09} obtain explicitly the two-sided exit time, the expected present value of tax until ruin, and the generalized Gerber-Shiu function.  Renaud \citeyear{R09} obtains explicit expressions of the distribution of the tax payments made over the lifetime of the company.

%The stochastic area till the first drawdown time represents the accumulated firm value till the drawdown event, and reflects the profitability and solvency of a firm during a financial crisis. The Azema-Yor stopping time is a generalized drawdown time and is a candidate solution to many optimal stopping problems (Graversen and Peskir \citeyear{GP97}, Pedersen \citeyear{ P05}, Shepp and Shiryaev \citeyear{SS93}), and also to the Skorokhod embedding problem (Obloj \citeyear{O04}).

 %We also consider the stochastic area till an Azema-Yor stopping time.

We make three contributions to the current literature. First, we obtain the Laplace transform of the occupation time of an Azema-Yor process below a constant level until the first hitting time of another Azema-Yor process or until an independent exponential time. This result unifies and extends recent literature(Li, Tang and Zhou \citeyear{LTZ13}, Li and Zhou \citeyear{LZ13} and Zhang \citeyear{Z13}) incorporating some of their results as special cases.
Second, we propose the ``Omega risk model with tax" to model the ruin and bankruptcy of an insurance company. This allows a more practical view in the modeling of bankruptcy, because an insurance company under distress is subject to tax, which may further weaken their solvency, and the company is considered bankrupt only when its surplus value is below a critical level beyond a ``grace period". We explicitly characterize the Laplace transform of the time of bankruptcy. Third, as an application of the main results, we obtain the Laplace transforms of the occupation times related to both the (absolute) drawdown and the \textit{relative} drawdown until respectively the first hitting time or an independent exponential time. We also discuss an extension to integral functionals through stochastic time change.
%Third, through stochastic time change, we the Omega risk model with tax and a general bankruptcy rate function $\omega(x)$. Previous literature considers the special case when $\omega(x)$ is a constant or a piecewise constant (Albrecher, Gerber and Shiu \citeyear{AGS11}, Gerber, Shiu and Yang \citeyear{GSY12}, and Li and Zhou \citeyear{LZ13}). We extend the literature to a general bankruptcy rate function and this provides more flexibility in modeling since $\omega(x)$ should be a decreasing function to comply with practical features of the model.
%Third, we prove some interesting identities between the risk models with and without tax, and obtain a number of invariant functionals and random times. Examples include the last passage time of the maximum for the diffusion with tax before the drawdown time, the speed of insurance reserve depletion(Zhang and Hadjiliadis \citeyear{ZH12}, Morales ), the first range time of diffusion with tax, and the frequency of drawdowns of diffusion with tax(Landriault, Li and Zhang \citeyear{LLZ14}).
%These invariance identities also hold for the case of a spectrally negative Levy process(SNLP).

%Third, we explicitly compute the expected ruin time in a diffusion risk model with value-dependent tax rate.

The paper is organized as follows: Section \ref{s1} reviews the preliminary results on the Omega risk model and the risk model with tax.  Section \ref{s2}
 gives the main results, namely the explicit Laplace transforms of the occupation time of an Azema-Yor process below a constant level until the first hitting time of another Azema-Yor process or until an independent exponential time. As an application, we propose the ``Omega risk model with tax", and determine the Laplace transform of the time of bankruptcy. We also discuss other interesting applications involving both the absolute and relative drawdown processes of the before-tax and after-tax processes, and the extension to a more general bankruptcy function. Section \ref{s3} provides examples using a standard Brownian motion with drift. Section \ref{s4} concludes the paper with future research directions.

\section{Preliminaries\label{s1}}

Recently, there are two strands of literature with one looking at a new definition of ``ruin", and the other considering a diffusion risk model refracted by its running maximum named the ``risk model with tax". We review relevant literature here and in Section \ref{s2} we will combine them to propose and study the ``Omega risk model with tax".

\subsection{The Omega risk model }
%Reference: Zhang \citeyear{Z13}, Li and Zhou \citeyear{LZ13}.

Classical ruin theory assumes that ruin or bankruptcy will occur at the first time when the surplus value of an insurance company is negative. For a pointer to the literature in this area, please refer to Gerber and Shiu \citeyear{GS98}. Recently, the ``Omega risk model" has been proposed and studied in a series of papers starting with Albrecher, Gerber and Shiu \citeyear{AGS11}. This model distinguishes between \textit{ruin} (negative surplus value) and \textit{bankruptcy} (going out of business). The company continues operation even with a period of negative surplus value, and is declared bankrupt if this period exceeds a threshold ``grace period". They introduce a bankruptcy rate function $\omega(x)$, where $x< 0$ denotes the value of negative surplus value, and it represents the probability of bankruptcy within $dt$ time units. The Omega risk model is based on the study of the occupation time of the risk process below a constant level. The occupation time of of a spectrally negative Levy process has been studied in Landriault, Renaud and Zhou (\citeyear{LRZ11} \citeyear{LRZ15}) and Loeffen, Renaud and Zhou \citeyear{LRZ14}. The occupation time of a refracted Levy process(Kyprianou and Loeffen \citeyear{KL10}) has been studied in Renaud \citeyear{R14}. This paper focuses on the diffusion risk model similar as in Li and Zhou \citeyear{LZ13}.

Given a complete filtered probability space $(\Omega, \mathcal{F}, \mathcal{F}_t, P)$ with state space $J=(l,\infty), -\infty\leq l <\infty$,  consider a $J$-valued regular time-homogeneous diffusion $X=(X_t)_{t\in[0,\infty)}$ which satisfies the stochastic differential equation(SDE)
\begin{equation}
dX_t=\mu(X_t)\, dt\, +\, \sigma(X_t)\,  dW_t, \quad X_0=x\in J, \label{G}
\end{equation}
where $W$ is a $\mathcal{F}_t$-Brownian motion and $\mu(\cdot)$ and $\sigma(\cdot)>0$ are Borel functions satisfying the following conditions: there exists a constant $C>0$ such that, for all $x_1, x_2\in J$
\begin{align}
\mid \mu(x_1)-\mu(x_2)\mid +\mid \sigma(x_1)-\sigma(x_2)\mid &\leq C\mid x_1-x_2\mid, \quad \mu^2(x_1)+\sigma^2(x_1)\leq C^2 (1+x_1^2), \label{cond1}
\end{align}
 Condition \eqref{cond1} guarantees that the SDE \eqref{G} has a unique solution that possesses the strong Markov property (see p.$40$, p.$107$, Gihman and Skorohod \citeyear{GS72}).

In the following, we denote $P_x(\cdot)\triangleq P(\cdot\mid X_0=x)$ and $E_x[\cdot]\triangleq E_x[\cdot\mid X_0=x]$.
Assume that the before-tax value of the company is modeled by $X$ with SDE \eqref{G}. If we introduce an auxiliary ``bankruptcy monitoring" process $N$ on the same probability space (with a possibly enlarged filtration to accommodate it), and assume that conditional on $X$, $N$ follows a Poisson process with state-dependent intensity $\omega(X_t) \1_{\lbrace X_t<0\rbrace }, t>0$. Define the time of bankruptcy $\tau_{\omega}$ as the first arrival time of the Poisson process $N$, i.e.
\begin{align}
\tau_{\omega}:=\inf\left\{t\geq 0: \int_0^t \omega(X_s)\1_{\lbrace X_s<0\rbrace}ds >e_1  \right\},
\end{align}
where $e_1$ is an independent exponential random variable with unit rate. Similar as in Li and Zhou \citeyear{LZ13}, for $\lambda>0$, we can express the Laplace transform of the time of bankruptcy as
\begin{align}
E_{x}[ e^{-\lambda \tau_{\omega}}]&=P_{x}(\tau_{\omega}<e_{\lambda})=1-E_{x}\left[ e^{-\int_0^{e_{\lambda}} \omega(X_s)\1_{\lbrace X_s<0\rbrace}ds } \right].\label{lz3}
\end{align}
% Similar as in Gerber, Shiu and Yang \citeyear{GSY12}, define the (total) \textit{exposure} as
%\begin{align}
%\mathcal{E}&:=\int_0^{\infty} \omega(X_s)\1_{\lbrace X_s<0\rbrace}ds.\label{ar}
%\end{align}
%
%
%Let $\lambda\rightarrow 0+$ in \eqref{lz3}, we can calculate the \textit{probability of bankruptcy} as
%\begin{align}
%\psi(x)&=P_{x}(\tau_{\omega}<\infty)=1-E_{x}[e^{-\mathcal{E}}],\label{st}
%\end{align}

\subsection{Risk model with surplus-dependent tax}

%Azema and Yor \citeyear{AY79} introduced a family of simple local martingales and proposed a solution to the Skorokhod embedding problem. These processes are later named Azema-Yor processes and the associated passage time is called the Azema-Yor stopping time. Their applications range from solving the Skorokhod embedding problem (Obloj \citeyear{O04}), pricing capped Russian options (Ott \citeyear{O13}) and portfolio optimization with drawdown constraints (El Karoui and Meziou \citeyear{EM06})
%
%%Consider a regular time-homogeneous diffusion $V$ in \eqref{G}, and denote $\overline{X}_t:=\max_{0\leq u\leq t}V_u$ its running maximum. The Azema-Yor stopping time is
%Consider the diffusion $V$ as defined in \eqref{G} with initial value $x>0$, and define the running maximum process associated with $V$ as
%\begin{align}
%\overline{X}_t &=(\max_{0\leq u \leq t} V_u)\vee s.\label{rm}
%\end{align}
%started at $s\geq x >0$. Define the Azema-Yor stopping time as
%\begin{align}
%\tau_{AY} &=\inf \{t>0:  X_t \leq g(\overline{X}_t)  \},
%\end{align}
%for any continuous function $g$ defined on $[0,\infty)$ satisfying $0<g(x)<x$ for $x>0$.

The risk model with tax was introduced by Albrecher and Hipp \citeyear{AH07} in the case of a constant tax rate, and was later extended by Albrecher, Renaud and Zhou \citeyear{ARZ08} and Kyprianou and Zhou \citeyear{KZ09} to the case where there is a non-negative state-dependent tax payment paid \textit{immediately} when the surplus value of the company is at a running maximum.
%For a pointer to the literature in this area, please refer to Albrecher, Renaud and Zhou \citeyear{ARZ08} and the references therein.

%We cast our model in a regular time-homogeneous diffusion setting (Li, Tang and Zhou \citeyear{LTZ13}).
Assume that the before-tax value of the company is modeled by the diffusion $X$ in \eqref{G}. Introduce a state-dependent tax: whenever the process $X_t$ coincides with its running maximum $\overline{X}_t$, the firm pays tax at rate $\gamma(\overline{X}_t)$, where $\gamma(\cdot):[x, \infty)\rightarrow [0,1)$ is a Borel measurable function. The value process after taxation is denoted as $(U_t)_{t\geq 0}$, and satisfies
\begin{align}
dU_t &=dX_t -\gamma(\overline{X}_t) d\overline{X}_t, \quad  U_0 =X_0 =x, \quad t\geq 0,\label{U}
\end{align}
%
%For a constant default threshold $a$ (conventionally assigned $0$ in the ruin theory), define the \textit{time of default with tax} as
%\begin{align}
%T^U (a) &=\inf\{t\geq 0 : U_t \leq a\},
%\end{align}
%and $\inf \emptyset =\infty$ by convention. Note that $a< x$. Now we want to compute $E[T^U (a)]$.

Kyprianou and Zhou \citeyear{KZ09} introduce the following function
\begin{align}
\overline{\gamma}(u)&=u-\int_{x}^u \gamma(z)dz=x +\int_{x}^u (1-\gamma(z))dz, \quad u> x .
\end{align}
Notice that $x< \overline{\gamma}(u)\leq u$. We have the following representation $U_t =X_t -\overline{X}_t +\overline{\gamma}(\overline{X}_t)$.

%Li, Tang and Zhou \citeyear{LTZ13} define the \textit{time of default with tax} as the two-sided exit time of $U$ from the constant boundary $[a,b]$,
%and compute explicitly the Laplace transforms. In Section \ref{s2}, we extend it to the occupation time of $U$ below a constant level exceeding a time threshold(\textit{grace period}).

%Define$T^U (a) =\inf\{t\geq 0 : U_t \leq a\},$
%$
%g(u)=u-\overline{\gamma}(u)=\int_{x}^u \gamma(z)dz.%\label{ge}
%$
%We have that $0\leq g(u)<u-x<u$ because $\gamma(\cdot):[x, \infty)\rightarrow [0,1)$.

\section{Omega risk model with surplus-dependent tax\label{s2}}
%Previously when we consider the Omega risk model where the surplus value of an insurance company is modeled by \eqref{G}. Now we consider a more realistic model that takes into account the taxation perspective of a company.

We combine the practical features of the ``Omega risk model" and the ``risk model with tax" to propose the ``Omega risk model with tax", where we use \eqref{U} to model the after-tax surplus value of an insurance company.  Li, Tang and Zhou \citeyear{LTZ13} define the \textit{time of default with tax} as the two-sided exit time of $U$ from the constant boundaries. We define the \textit{time of bankruptcy} of the company as the first time the occupation time exceeds an independent exponential time with unit rate.

%the two-sided exit problem of the surplus value process as the time of bankruptcy(Li, Tang and Zhou \citeyear{LTZ13}), we consider the occupation time of the process $U_t$ below a constant level until a first hitting time or an independent exponential time.

$U_t$ is a special case of the so called Azema-Yor process introduced in Azema and Yor \citeyear{AY79}, which is a process refracted by functionals of its running maximum (see also the terminology in Albrecher and Ivanovs \citeyear{AI14}). We first obtain general Laplace transforms of the occupation time of an Azema-Yor process below a constant level until the first hitting time of another Azema-Yor process or until an independent exponential time. To the best of our knowledge, these results are new and are of independent interest.
 As an application, we obtain the explicit Laplace transform of the ``time of bankruptcy" of the ``Omega risk model with tax". Our general formula contains some results in Li and Zhou \citeyear{LZ13} and Zhang \citeyear{Z13} as special cases.
 \begin{re}
 Note that $U$ defined in \eqref{U} is a special case of the general Azema-Yor process introduced below, so our strategy is to first study the occupation time of a general Azema-Yor process, and then specialize to the after-tax process in \eqref{U}. In the following, we use the same notation $U_t$ to denote a general Azema-Yor process and we shall mention explicitly whenever we refer to the process in \eqref{U}.
\end{re}

Consider the following two general Azema-Yor processes:
\begin{align}
V_t&:=X_t-h(\overline{X}_t); \quad \quad U_t:=X_t-g(\overline{X}_t), \quad V_0=U_0=x,\label{ay}
\end{align}
where $h$ and $g$ are defined on $[x,\infty)$ satisfying $0\leq h(u)\leq u-x, 0\leq g(u)\leq u-x$ and $h(x)=g(x)=0$. Note that $V_t$ and $U_t$ are both constructed using $(X_t,\overline{X}_t)$, but they may have possibly different $h(\cdot)$ and $g(\cdot)$. If $h(\cdot)=g(\cdot)$, then $V_t=U_t$, $P$-a.s., $t>0$.
%
%\subsection{Occupation time of $U$ until the exit time of $V$}
%

%Note that there are two Azema-Yor functionals under consideration here and they are both built from $(X_t, \overline{X}_t)$. One is $U_t= X_t-g(\overline{X}_t)$. The other is $V_t=X_t-h(\overline{X}_t)$.

In the following,  fix two constants $y$ and $a$ such that $-x\leq y<a$. Define $y^{\prime}=y+x$ and $a^{\prime}=a+x$, which satisfy $0\leq y^{\prime}<a^{\prime}$, and are useful later when we compare our results to those of Zhang \citeyear{Z13}.
Define the first hitting time of $V$ to $-a$ as
\begin{align}
\tau_{h, a}:=\inf\left\lbrace t> 0: V_t\leq -a     \right\rbrace=\inf\left\lbrace t> 0: h(\overline{X}_t)-X_t\geq a     \right\rbrace.\label{tauh}
\end{align}
%Note that is just the Azema-Yor stopping time with a function $h(\cdot)$.

We introduce some notations that are consistent with Zhang \citeyear{Z13} which will be used later in the proof.
Define $\tau_m^{\pm}:=\inf\lbrace t>0: X_t \gtreqqless m \rbrace, m\in J$, and define $\phi^{+}_q(\cdot)$ and $\phi^{-}_q(\cdot)$ respectively as the increasing and decreasing positive solutions of the \textit{Sturm-Liouville} ordinary differential equation $\frac{1}{2}\sigma^2(x) f^{\prime\prime}(x)+\mu(x)f^{\prime}(x)=qf(x)$. If we fix the scale function of $X$ as $s(\cdot)$, then there exists a positive constant $w_q$ such that $w_q s^{\prime}(x)=(\phi^{+}_q)^{\prime}(x) \phi^{-}_q(x)-(\phi^{+}_q)^{\prime}(x)\phi^{+}_q(x)$. Define the auxiliary functions $W_{q}(x,y):=\frac{1}{w_q}(\phi^{+}_q(x) \phi^{-}_q(y)-\phi^{+}_q(y)\phi^{-}_q(x))$, $W_{q,1}(x,y):=\frac{\partial}{\partial x}W_q(x,y)$ and $W_{q,2}(x,y):=\frac{\partial}{\partial y}W_{q,1}(x,y)$.

The main object of interest is the occupation time of $U_t$ below $-y$ until $\tau_{h, a}$:
\begin{align}
G_y^{a,h, g}&:=\int_0^{\tau_{h, a}} \1_{\lbrace U_t<-y \rbrace}dt.\label{ob}
\end{align}
%By choosing different combinations of $(h(\cdot), g(\cdot))$, we can show that some results in the recent literature are just special cases of the general result in the Theorem.

%Before we study the Laplace transform of $G_y^{a,h, g}$,

The following is a slight generalization of Proposition 1 of Zhang and Hadjiliadis \citeyear{ZH12}, which studies the path decomposition of $U_t$ for $t\in [0,\tau_{h, a}]$.  Define the first drawdown time of $X$ as $\sigma_a:=\inf\lbrace t>0: \overline{X}_t-X_t\geq a \rbrace$. If $g(u)=0, h(u)=u-x$, then  $U_t=X_t$ and $\tau_{u, a}=\sigma_{a^{\prime}}$, $P$-a.s., and the following result reduces to Proposition 1 of  Zhang and Hadjiliadis \citeyear{ZH12} by substituting $a^{\prime}\rightarrow K$ there.

 %Note that they consider the path decomposition on $[0, \tau^D_a]$, where $\tau^D_a$ is the first drawdown time of the $X$ process of $a$ units.

\begin{proposition}(Path decomposition of $U$ until an Azema-Yor stopping time, generalization of Proposition 1 of Zhang and Hadjiliadis \citeyear{ZH12})\label{de}

With $\tau_{h, a}$ defined in \eqref{tauh}, consider the last passage time of $X$ to its running maximum before $\tau_{h, a}$:
\begin{align}
\rho:=\sup\left\lbrace t\in[0, \tau_{h, a}]: X_t =\overline{X}_t \right\rbrace.
\end{align}
Conditional on $X_{\rho}$, the path fragments $\lbrace U_t\rbrace_{t\in[0, \rho]}$ and $\lbrace U_t\rbrace_{t\in[\rho, \tau_{h, a}]}$ are two independent processes.

Denote $Y_t^{\rho}:=P_x(\rho>t\mid \mathcal{F}_t)$.
Then $Y_t^{\rho}$ is a supermartingale and has the Doob-Meyer decomposition
\begin{align}
Y_t^{\rho}&=M_t^{\rho}-L_t^{\rho},\label{dm}
\end{align}
where
\begin{align}
Y_t^{\rho}&=P_x(\rho>t \mid \mathcal{F}_t)=\frac{s(X_t)-s(h(\overline{X_t})-a)}{s(\overline{X}_t)-s(h(\overline{X_t})-a)} \1_{\lbrace t<\tau_{h, a}\rbrace},\label{dmy}
\end{align}
\begin{align}
M_t^{\rho}&=1+\int_0^{t\wedge \tau_{h, a}  } \frac{s^{\prime}(X_u)\sigma(X_u)}{s(\overline{X}_u)-s(h(\overline{X}_u)-a)}dW_u,\label{dmm}
\end{align}
and
\begin{align}
L_t^{\rho}&=\int_0^{t\wedge \tau_{h, a} } \frac{s^{\prime}(\overline{X}_u)}{s(\overline{X}_u)-s(h(\overline{X}_u)-a)} d\overline{X}_u.\label{dml}
\end{align}
\end{proposition}
 \begin{proof}
The proof is similar to that of Proposition 1, p.744 of Zhang and Hadjiliadis \citeyear{ZH12}, but some steps need non-trivial adjustments. Thus we present the proof here for completeness.

Note that $\lbrace \rho>t \rbrace$ means that, $\lbrace t<\tau_{h, a}\rbrace$ holds and the path of $X$ will revisit $\overline{X}_t$ before it reaches $h(\overline{X}_t)-a$. So we have
\begin{align}
Y_t^{\rho}&=P_x(\rho>t \mid \mathcal{F}_t)=\frac{s(X_t)-s(h(\overline{X_t})-a)}{s(\overline{X}_t)-s(h(\overline{X_t})-a)} \1_{\lbrace t<\tau_{h, a}\rbrace},
\end{align}
 %where $s(\cdot)$ is the scale function of $X$.
 For any $t\in[0, \tau_{h, a})$, apply Ito's lemma
 \begin{align}
 dY_t^{\rho}&=\frac{d[s(X_t)-s(h(\overline{X_t})-a)]}{s(\overline{X}_t)-s(h(\overline{X_t})-a)}-\frac{s(X_t)-s(h(\overline{X_t})-a)}{[s(\overline{X}_t)-s(h(\overline{X_t})-a)]^2}d[s(\overline{X}_t)-s(h(\overline{X_t})-a)].\label{310}
 \end{align}
 We have the following intermediate calculations: $d[s(X_t)-s(h(\overline{X_t})-a)]=s^{\prime}(X_t)\sigma(X_t)dW_t-s^{\prime}(h(\overline{X_t})-a)h^{\prime}(\overline{X}_t))d\overline{X}_t$ and $d[s(\overline{X}_t)-s(h(\overline{X_t})-a)]=[s^{\prime}(\overline{X}_t)-s^{\prime}(h(\overline{X_t})-a)h^{\prime}(\overline{X}_t))]d\overline{X}_t$. Note that the measure $d\overline{X}_t$ is supported on $\lbrace t\mid X_t=\overline{X}_t\rbrace$. The above two expressions combined with \eqref{310} lead to
 \begin{align}
 dY_t^{\rho}&=\frac{s^{\prime}(X_t)\sigma(X_t)}{s(\overline{X}_t)-s(h(\overline{X_t})-a)}dW_t-\frac{s^{\prime}(\overline{X}_t)}{s(\overline{X}_t)-s(h(\overline{X_t})-a)}d\overline{X}_t.\label{311}
 \end{align}
 Integrate \eqref{311} from $0$ to $t\in[0,\tau_{h,a})$ and note that $Y_0^{\rho}=1$ and $\lim\limits_{t\uparrow \tau_{h,a}}Y_t^{\rho}=0$, then \eqref{dmm} and \eqref{dml} follow.

 We can derive similar results as Proposition 2 and Proposition 4 of Zhang and Hadjiliadis \citeyear{ZH12}. In particular, conditionally on $X_{\rho}=m$, $\lbrace X_t\rbrace_{t\in[0,\rho]}$ has the same law as the unique weak solution of the following SDE stopped at the first hitting time of level $m$
 \begin{align}
 dZ_t &= \(\mu(Z_t) +\frac{s^{\prime}(Z_t)\sigma^2(Z_t)}{s(Z_t)-s(h(\overline{Z}_t)-a)}   \)dt+\sigma(Z_t) dB_t, \quad Z_0=x.
 \end{align}

 Conditionally on $X_{\rho}=m$, $\lbrace m-X_t\rbrace_{t\in[\rho, \tau_{h,a}]}$ has the same law as the unique weak solution of the following SDE stopped at the first hitting time of level $a$
 \begin{align}
 dJ_t &= \(-\mu(m-J_t) +\frac{s^{\prime}(m-J_t)\sigma^2(m-J_t)}{s(m)-s(h(m)-J_t)}   \)dt-\sigma(m-J_t) dB_t, \quad J_0=0.
 \end{align}
 Similar to the proof of Proposition 4 of Zhang and Hadjiliadis \citeyear{ZH12}, we have that $\lbrace X_t\rbrace_{t\in [\rho, \tau_{h, a}]}$ and $\lbrace X_t\rbrace_{t\in[0, \rho]}$, or equivalently $\mathcal{F}_{\rho}$, are conditionally independent. Conditionally on $\overline{X}_{\rho}=m$,  for $t\in[\rho, \tau_{h, a}]$, we have $\overline{X}_t =\overline{X}_{\rho}=m$, and  $U_t=X_t-g(\overline{X}_t)=X_t-g(m)$, $P$-a.s. Then we have that $\lbrace U_t\rbrace_{t\in [\rho, \tau_{h, a}]}$ and $\mathcal{F}_{\rho}$ are conditionally independent. For $t\in[0,\rho]$, $U_t=X_t-g(\overline{X}_t)$ is adapted to $\mathcal{F}_{\rho}$, thus $\lbrace U_t\rbrace_{t\in[0, \rho]}$ and $\lbrace U_t\rbrace_{t\in[\rho, \tau_{h, a}]}$ are two conditionally independent  processes. This completes the proof. \qed
 \end{proof}

 % Since  $\lbrace X_t\rbrace_{t\in [\rho, \tau_{h, a}]}$ and $\lbrace X_t\rbrace_{t\in[0, \rho]}$(equivalently $\mathcal{F}_{\rho}$) are conditionally independent(Proposition 4 of Zhang and Hadjiliadis \citeyear{ZH12}),

%We also recall the following useful result from Zhang \citeyear{Z13} using our notation.
%\begin{lemma}(Proposition 4.1 of Zhang \citeyear{Z13})
%
%\end{lemma}

Now we present the main result of this section: the Laplace transform of $G_y^{a,h,g}$.
\begin{theorem}(Occupation time until first hitting for two Azema-Yor processes, generalization of Theorem 4.5 of Zhang \citeyear{Z13})\label{mt}

For $q\geq 0$, $-x\leq y<a$, if $g(u)\geq h(u)$ for $u\in[x,\infty)$, then we have
\begin{align}
E_x[ e^{-q G_y^{a,h, g}}; \tau_{h, a}<\infty ]&= \int_x^{\infty} \frac{ \frac{s^{\prime}(m)}{W_q(g(m)-y, h(m)-a)}    }{1+\frac{s(m)-s(g(m)-y)}{s^{\prime}(g(m)-y)}\frac{W_{q,1}(g(m)-y, h(m)-a)}{W_q(g(m)-y, h(m)-a)}}\notag\\
&\quad \times \exp\left(-\int_x^m \frac{ \frac{s^{\prime}(u)}{s^{\prime}(g(u)-y)} \frac{W_{q,1}(g(u)-y, h(u)-a  )}{W_q(g(u)-y, h(u)-a  )}      }{1+ \frac{s(u)-s(g(u)-y)}{s^{\prime}(g(u)-y)}  \frac{W_{q,1}(g(u)-y, h(u)-a  )}{W_q(g(u)-y, h(u)-a  )}   } du   \right)    dm.\label{main}
\end{align}
\end{theorem}

\begin{proof}
The proof is similar to that of Theorem 4.5 of Zhang \citeyear{Z13}, but needs some non-trivial adaptations where needed. We present the proof for completeness.
Introduce a non-negative bounded optional process
\begin{align*}
I_t &=\exp\left(-q\int_0^t \1_{\lbrace U_s<-y\rbrace}ds\right)\1_{\lbrace t<\tau_{h, a}<\infty\rbrace}, \quad t\geq 0.
\end{align*}

From Theorem 15, p.380 of Protter \citeyear{P04} combined with the decomposition in Proposition \ref{de}, we have that for any positive test function $f(\cdot)$ on $[0,\infty)$
\begin{align*}
E_x[f(X_{\rho}) I_{\rho} ]&=E_x \left[ \int_0^{\infty} f(X_t)I_t dL_t \right]=E_x \left[ \int_0^{\infty} \frac{f(X_t)I_t s^{\prime}(\overline{X}_t)}{s(\overline{X}_t)-s(h(\overline{X_t})-a)} d\overline{X}_t \right].
\end{align*}

Apply a change of variable $m=\overline{X}_t$, and recall that the measure $d\overline{X}_t$ is supported on $\lbrace t\mid X_t=\overline{X}_t\rbrace$, then
\begin{align}
E_x[f(X_{\rho}) I_{\rho} ]&=\int_x^{\infty} f(m) E_x \left[ \exp\left(-q\int_0^{\tau_m^{+}} \1_{\lbrace U_s<-y\rbrace}ds\right)\1_{\lbrace \tau_m^{+}< \tau_{h, a}\rbrace} \right] \frac{s^{\prime}(m)}{s(m)-s(h(m)-a)}dm. \label{24}
\end{align}

From equation (20) on p.607 in Lehoczky \citeyear{L77} with the substitution $v-h(v)+a\rightarrow u(v)$, we have
\begin{align}
P_x(\tau_m^{+}< \tau_{h, a})&= \exp \left( -\int_x^m \frac{s^{\prime}(v)}{s(v)-s(h(v)-a)}dv \right),\label{25}
\end{align}
\begin{align}
P_x (X_{\rho} \in dm)%&=P_x (\overline{X}_{\rho} \in dm)=-dP_x(\tau_m^{+}< \tau_{h, a})\notag\\
&=\frac{s^{\prime}(m)}{s(m)-s(h(m)-a)} \exp\left(-\int_x^m \frac{s^{\prime}(v)}{s(v)-s(h(v)-a)}dv\right).\label{26}
\end{align}

From \eqref{24} and \eqref{25}, we have
\begin{align}
E_x[ f(X_{\rho})I_{\rho})  ]&=\int_x^{\infty} f(m) E_x\left[  \exp\left(-q\int_0^{\tau_m^{+}} \1_{\lbrace U_t<-y\rbrace}dt\right)\mid \tau_m^{+}<\tau_{h, a}\right]\notag\\
&\quad \times \frac{s^{\prime}(m)}{s(m)-s(h(m)-a)} \exp\left(-\int_x^m \frac{s^{\prime}(v)}{s(v)-s(h(v)-a)}dv\right)dm\notag\\
&=\int_x^{\infty} f(m) E_x\left[  \exp\left(-q\int_0^{\tau_m^{+}} \1_{\lbrace U_t<-y\rbrace}dt\right)\mid \tau_m^{+}<\tau_{h, a}\right] P_x (X_{\rho} \in dm).\label{h1}
\end{align}
On the other hand
\begin{align}
E_x[ f(X_{\rho})I_{\rho})  ]
&=\int_x^{\infty} f(m) E_x\left[  \exp\left(-q\int_0^{\rho}  \1_{\lbrace U_t<-y\rbrace}dt  \right)\mid X_{\rho}=m \right] P_x (X_{\rho} \in dm).\label{h2}
\end{align}
From \eqref{h1} and \eqref{h2} and the arbitrariness of $f(\cdot)$, we have
\begin{align*}
E_x \left[ \exp\left(-q\int_0^{\rho}  \1_{\lbrace U_t<-y\rbrace}dt  \right)\mid X_{\rho}=m \right]&=E_x\left[  \exp\left(-q\int_0^{\tau_m^{+}} \1_{\lbrace U_t<-y\rbrace}dt\right)\mid \tau_m^{+}<\tau_{h, a}\right]
\end{align*}

 %The following is  similar to the corresponding part of the proof of Theorem 4.5 in Zhang \citeyear{Z13} except that we consider the event $\1_{\lbrace U_t<-y \rbrace}=\1_{\lbrace X_t<g(\overline{X}_t)-y \rbrace}$.
Define $A_y^{a,b}:=\int_0^{\tau_a^{-} \wedge \tau_b^{+}} \1_{\lbrace X_t<y\rbrace}dt$, and also $\varepsilon=(m-x)/N$ for a large integer $N>0$. For $i=0,1,...,N-1$, when $X$ starts at $x+i\varepsilon$, the condition $\lbrace\tau_m^{+}<\tau_{h, a}\rbrace$ requires that at each time, the process shall hit the level $x+(i+1)\varepsilon$ before its hits $h(x+i\varepsilon)-a$.
 From the Lebesgue dominated convergence theorem, continuity and the strong Markov property of $X$, we have
\begin{align}
&E_x\left[  \exp\left(-q\int_0^{\tau_m^{+}} \1_{\lbrace U_t<-y\rbrace}dt\right)\mid \tau_m^{+}<\tau_{h, a}\right]\notag\\
&=\lim\limits_{N\rightarrow \infty} \prod\limits_{i=0}^{N-1} E_{x+i\varepsilon} \left[e^{-q A^{h(x+i\varepsilon)-a, x+(i+1)\varepsilon}_{g(x+i \varepsilon) -y}} \mid \tau^{+}_{x+(i+1)\varepsilon} <\tau^{+}_{h(x+i\varepsilon)-a}   \right]\notag\\
&=\lim\limits_{N\rightarrow \infty} \exp\left( \log \left( \sum_{i=0}^{N-1} E_{x+i\varepsilon} \left[e^{-q A^{h(x+i\varepsilon)-a, x+(i+1)\varepsilon}_{g(x+i \varepsilon) -y}} \mid \tau^{+}_{x+(i+1)\varepsilon} <\tau^{+}_{h(x+i\varepsilon)-a}   \right]          \right)   \right)\notag\\
&=\exp\left( \lim\limits_{N\rightarrow \infty} \sum_{i=0}^{N-1} \left(E_{x+i\varepsilon} \left[e^{-q A^{h(x+i\varepsilon)-a, x+(i+1)\varepsilon}_{g(x+i \varepsilon) -y}} \mid \tau^{+}_{x+(i+1)\varepsilon} <\tau^{+}_{h(x+i\varepsilon)-a}   \right]-1\right) \right)\notag\\
&=\exp\left(\int_x^m \left[ \frac{s^{\prime}(u)}{s(u)-s(h(u)-a)}-\frac{ \frac{s^{\prime}(u)}{s^{\prime}(g(u)-y)} \frac{W_{q,1}(g(u)-y, h(u)-a  )}{W_q(g(u)-y, h(u)-a  )}      }{1+ \frac{s(u)-s(g(u)-y)}{s^{\prime}(g(u)-y)}  \frac{W_{q,1}(g(u)-y, h(u)-a  )}{W_q(g(u)-y, h(u)-a  )}   }      \right]du    \right),\label{key1}
\end{align}
and the last equality follows from the second expression of equation (18) in Proposition 4.1 of Zhang\citeyear{Z13}. This is because $h(x+i\varepsilon )-a < g(x+i \varepsilon) -y \leq x+i\varepsilon-x-y\leq  x+i\varepsilon <x+(i+1)\varepsilon$ for all $-x\leq y<a$. Note that both expressions in equation (18) of Zhang \citeyear{Z13} agree at the boundary case, thus we can include the case when $y=-x$ and still apply the second expression of equation (18) to proceed.

For the occupation time on $[\rho, \tau_{h,a}]$, we have
\begin{align}
&E_x\left[\exp\left(-q\int_{\rho}^{\tau_{h,a}} \1_{\lbrace U_t<-y\rbrace} dt  \right) \mid X_{\rho} =m\right]\notag\\
&=E_m \left[\exp\left(-q\int_{0}^{\tau^{-}_{h(m)-a}} \1_{\lbrace X_t<g(m)-y\rbrace} dt  \right) \mid \tau^{-}_{h(m)-a} <\tau_{m}^{+}\right]\notag\\
&=\lim\limits_{\delta\rightarrow 0+} \frac{E_m \left[e^{-q A_{g(m)-y}^{h(m)-a, m+\delta}}; \tau^{-}_{h(m)-a} <\tau_{m+\delta}^{+}\right]}{P_m(\tau^{-}_{h(m)-a} <\tau_{m+\delta}^{+})}\notag\\
&=\frac{\frac{s(m)-s(h(m)-a)}{W_q(g(m)-y, h(m)-a  )}   }{1+\frac{s(m)-s(g(m)-y)}{s^{\prime}(g(m)-y)}\frac{W_{q,1}(g(m)-y, h(m)-a)}{W_q(g(m)-y, h(m)-a)}},\label{key2}
\end{align}
and the last equality follows from the second part of (19) of Zhang \citeyear{Z13}, because $h(m)-a<g(m)-y\leq m-x-y\leq m<m+\delta$.
From Proposition \ref{de}, $\lbrace U_t\rbrace_{t\in[0, \rho]}$ and $\lbrace U_t\rbrace_{t\in[\rho, \tau_{h, a}]}$ are two conditionally independent processes, thus we have \begin{align}
E_x\left[\exp\left(-q\int_{0}^{\tau_{h,a}} \1_{\lbrace U_t<-y\rbrace} dt  \right) \mid X_{\rho} =m\right]&=E_x\left[\exp\left(-q\int_{0}^{\rho} \1_{\lbrace U_t<-y\rbrace} dt  \right) \mid X_{\rho} =m\right]\notag\\
& \times
 E_x\left[\exp\left(-q\int_{\rho}^{\tau_{h,a}} \1_{\lbrace U_t<-y\rbrace} dt  \right) \mid X_{\rho} =m\right].\label{int}
 \end{align}
We use the density in \eqref{26} to integrate out \eqref{int} and this completes the proof. \qed
\end{proof}

\begin{re}
The assumption $g(u)\geq h(u)$ for $u\in[x,\infty)$ is not overly restrictive because it contains many interesting cases for applications. In particular, it does not restrict the form of $g(\cdot)$ if $h(\cdot)=g(\cdot)$. If $g(u)=h(u)=u-x$, then $\tau_{u-x,a}=\sigma_{a^{\prime}}$ and $\1_{\lbrace U_t<-y \rbrace}=\1_{\lbrace Y_t>y^{\prime} \rbrace}$, where $Y_t=\overline{X}_t-X_t$ is the drawdown process. Then $G_y^{a, u-x, u-x}=\int_0^{\sigma_{a^{\prime}}} \1_{\lbrace Y_t>y^{\prime} \rbrace}dt=:C_{y^{\prime}}^{a^{\prime}}$, which is  equation (21) of Zhang \citeyear{Z13} with the substitutions $y^{\prime}\rightarrow y$ and $a^{\prime}\rightarrow a$ there. Note that $0\leq y^{\prime}<a^{\prime}$ and our formula \eqref{main} reduces to the formula in Theorem 4.5 of Zhang \citeyear{Z13} with the above substitutions. As a sanity check, when $y=-x$ and $g(u)=u-x$, we have $G_0^{a,h,u-x}=\int_0^{\tau_{h,a}}\1_{\lbrace X_t<\overline{X_t} \rbrace}dt=\tau_{h,a}$, $P$-a.s, and for $q\geq 0$, our formula \eqref{main} reduces to
\begin{align}
E_x \left[e^{-q \tau_{h,a}} \right]&=\int_x^{\infty} \frac{s^{\prime}(m)}{W_q(m, h(m)-a)} \exp\left( -\int_x^m  \frac{W_{q,1}(u, h(u)-a)}{W_q(u, h(u)-a)} du\right)dm,\label{lo}
\end{align}
which agrees with formula (21) on p.601 in Lehoczky \citeyear{L77} with substitutions $0\rightarrow \alpha$, $q\rightarrow \beta$ and $m-h(m)\rightarrow u(m)$. If we further take $h(u)=u-x$, then \eqref{lo} reduces to Proposition 3.3 of Zhang \citeyear{Z13}.

If $g(u)=h(u)=0$, then $\tau_{0,a}=\tau^{-}_{-a}$, and $G_y^{a, 0, 0}=\int_0^{\tau^{-}_{-a}} \1_{\lbrace X_t<-y \rbrace}dt$, which represents the occupation time of the process $X$ below $-y$ until its first hitting time to $-a$ from above. From \eqref{main}, we have
%\begin{align}
%E_x[ e^{-q G_y^{a,0, 0}}; \tau_{-a}<\infty ]&= \int_x^{\infty} \frac{ \frac{s^{\prime}(m)}{W_q(y, a)}    }{1+\frac{s(m)-s(y)}{s^{\prime}(y)}\frac{W_{q,1}(y, a)}{W_q(y, a)}}\notag\\
%&\quad \times \exp\left(-\int_x^m \frac{ \frac{s^{\prime}(u)}{s^{\prime}(y)} \frac{W_{q,1}(y, a  )}{W_q(y, a  )}      }{1+ \frac{s(u)-s(y)}{s^{\prime}(y)}  \frac{W_{q,1}(y, a  )}{W_q(y, a  )}   } du   \right)    dm.
%\end{align}
%Now take the limit as $a\rightarrow \infty$ and $y\rightarrow 0$, we have
%\begin{align}
%E_x[ e^{-q \int_0^{\infty} \1_{\lbrace X_t<0 \rbrace}dt} ]&=\lim\limits_{a\rightarrow \infty} \int_x^{\infty} \frac{ \frac{s^{\prime}(m)}{W_q(0, a)}    }{1+\frac{s(m)-s(0)}{s^{\prime}(0)}\frac{W_{q,1}(0, a)}{W_q(0, a)}}\notag\\
%&\quad \times \exp\left(-\int_x^m \frac{ \frac{s^{\prime}(u)}{s^{\prime}(0)} \frac{W_{q,1}(0, a  )}{W_q(0, a  )}      }{1+ \frac{s(u)-s(0)}{s^{\prime}(0)}  \frac{W_{q,1}(0, a  )}{W_q(0, a  )}   } du   \right)    dm\notag\\
%&=\lim\limits_{a\rightarrow \infty} \frac{s(x) W_{q,1}(0,a)+W_q(0,a)}{s(x)W^2_{q,1}(0,a)+W_q(0,a)W_{q,1}(0,a)}
%\end{align}
\begin{align}
E_x[ e^{-q G_y^{a,0, 0}}; \tau^{-}_{-a}<\infty ]&= \int_x^{\infty} \frac{ \frac{s^{\prime}(m)}{W_q(-y, -a)}    }{1+\frac{s(m)-s(-y)}{s^{\prime}(-y)}\frac{W_{q,1}(-y, -a)}{W_q(-y, -a)}}\notag\\
&\quad \times \exp\left(-\int_x^m \frac{ \frac{s^{\prime}(u)}{s^{\prime}(-y)} \frac{W_{q,1}(-y, -a  )}{W_q(-y, -a  )}      }{1+ \frac{s(u)-s(-y)}{s^{\prime}(-y)}  \frac{W_{q,1}(-y, -a  )}{W_q(-y, -a  )}   } du   \right)    dm\notag\\
&=\int_x^{\infty} \frac{s^{\prime}(m)}{W_q(-y, -a)}     \frac{1}{1+A(s(m)-s(-y))} \frac{1+A(s(x)-s(-y))}{1+A(s(m)-s(y))}dm\notag\\
&=\frac{1}{A W_q(-y, -a)}-\frac{1}{A W_q(-y, -a)} \frac{1+A(s(x)-s(-y))}{1+A(s(\infty)-s(-y))}\notag\\
%&=\frac{s(\infty)-s(x)}{W_q(-y, -a) (1+A(s(\infty)-s(y)))}\notag\\
&=\frac{(s(\infty)-s(x))s^{\prime}(-y)}{(s(\infty)-s(-y))W_{q,1}(-y, -a)+s^{\prime}(-y)W_q(-y, -a)},\label{check}
\end{align}
where $A=\frac{W_{q,1}(-y,-a)}{s^{\prime}(-y) W_q (-y, -a)}$. Note that \eqref{check} agrees with the second expression in equation (19) of Proposition $4.1$ of Zhang \citeyear{Z13} by letting $b\rightarrow \infty$ and substituting $-a\rightarrow a$, and $-y\rightarrow y$ there.

\end{re}
%Take $h(u)=g(u)=u-\overline{\gamma}(u)$ and $a\rightarrow \infty$ in Theorem \ref{mt}, we have
%\begin{theorem}
%(Perpetual occupation time for an Azema-Yor process)\label{mtinf}
%
%For $q\geq 0$, $y\geq 0$, if $g(u)\geq h(u)$ for $u\in(0,\infty)$, then we have
%\begin{align}
%E_x[ e^{-q G_y^{\infty,g, g}}; \tau_{h, a}<\infty ]&= \int_x^{\infty} \frac{ \frac{s^{\prime}(m)}{W_q(g(m)-y, h(m)-a)}    }{1+\frac{s(m)-s(g(m)-y)}{s^{\prime}(g(m)-y)}\frac{W_{q,1}(g(m)-y, h(m)-a)}{W_q(g(m)-y, h(m)-a)}}\notag\\
%&\quad \times \exp\left(-\int_x^m \frac{ \frac{s^{\prime}(u)}{s^{\prime}(g(u)-y)} \frac{W_{q,1}(g(u)-y, h(u)-a  )}{W_q(g(u)-y, h(u)-a  )}      }{1+ \frac{s(u)-s(g(u)-y)}{s^{\prime}(g(u)-y)}  \frac{W_{q,1}(g(u)-y, h(u)-a  )}{W_q(g(u)-y, h(u)-a  )}   } du   \right)    dm.\label{maininf}
%\end{align}
%\end{theorem}
%\begin{proof}
%Take $h(u)=g(u)$ and let $a\rightarrow \infty$
%
%\qed
%\end{proof}

\subsection{Occupation time of the Azema-Yor process until an independent exponential time}
For $y\geq -x$, consider the occupation time of $U$ below $-y$ until an independent exponential time $e_q, q>0$:
$$O^{q,g}_y :=\int_0^{e_q}\1_{\lbrace U_t<-y\rbrace}dt.$$

\begin{theorem}\label{mp}
For all $p, q>0$ and $y\geq -x$
\begin{align}
E_{x}[e^{-p O^{q,g}_y}]&=1-\exp\left( -\int_x^{\infty} \frac{W_{q,2}(g(u)-y,u)+W_{q,1}(u, g(u)-y)\frac{\phi^{+\prime}_{q+p}(g(u)-y)}{\phi^{+}_{q+p}(g(u)-y)}}{W_{q,1}(g(u)-y,u)+W_{q}(u, g(u)-y)\frac{\phi^{+\prime}_{q+p}(g(u)-y)}{\phi^{+}_{q+p}(g(u)-y)}}du    \right)\notag\\
&\quad - \int_x^{\infty} \exp\left( -\int_x^m \frac{W_{q,2}(g(u)-y,u)+W_{q,1}(u, g(u)-y)\frac{\phi^{+\prime}_{q+p}(g(u)-y)}{\phi^{+}_{q+p}(g(u)-y)}}{W_{q,1}(g(u)-y,u)+W_{q}(u, g(u)-y)\frac{\phi^{+\prime}_{q+p}(g(u)-y)}{\phi^{+}_{q+p}(g(u)-y)}} du       \right)\notag\\
&\quad\quad\quad \times \frac{ \frac{p}{q+p} s^{\prime}(m) \frac{\phi^{+\prime}_{q+p}(g(m)-y)}{\phi^{+}_{q+p}(g(m)-y)} }{W_{q,1}(g(m)-y,m)+W_{q}(m, g(m)-y)\frac{\phi^{+\prime}_{q+p}(g(m)-y)}{\phi^{+}_{q+p}(g(m)-y)}}dm.\label{main2}
\end{align}
%where $g(x)=x-\overline{\gamma}(x)$.
%where $0\leq g(x)<x$.
\end{theorem}
\begin{proof}
We consider $\1_{\lbrace U_t<-y\rbrace}$ instead of $\1_{\lbrace Y_t>y\rbrace}$, and the proof is based on similar non-trivial adaptations of that of Theorem 4.7 in Zhang \citeyear{Z13} by substituting $g(u)-y\rightarrow u-y$ and $g(m)-y\rightarrow m-y$ throughout. Note that $g(x+i\varepsilon)-y\leq x+i\varepsilon-x-y \leq x+i\varepsilon <x+(i+1)\varepsilon$, for $i=0,1...,N-1$ and $\varepsilon>0$. Thus we can safely apply Corollary $4.2$ of Zhang \citeyear{Z13} at an intermediate step of the proof.  \qed
\end{proof}
\begin{re}
If $g(u)=u-x$, then $O^{q,u-x}_y=\int_0^{e_q} \1_{\lbrace Y_t>y^{\prime} \rbrace}dt=:E_{y^{\prime}}^q$ as defined in equation (29) of Zhang \citeyear{Z13}. In this case our formula \eqref{main2} reduces to equation $(30)$ of his Theorem $4.7$ with substitution $y^{\prime}\rightarrow y$ there. As a sanity check, if $y=-x$ and $g(u)=u-x$, formula \eqref{main2} reduces to $E_{x}[e^{-p O^{q,u}_0}]=\frac{q}{q+p}$. On the other hand, $E_{x}[e^{-p O^{q,u-x}_{-x}}]=E_{x}[e^{-p  \int_0^{e_q}\1_{\lbrace X_t<\overline{X}_t \rbrace}dt}]=E_{x}[e^{-p e_q}]=\frac{q}{q+p}$.
\end{re}

If  $g(u)=0$, then for $y\geq -x$, $O^{q,0}_{-x}=\int_0^{e_q} \1_{\lbrace X_t<-y \rbrace}dt$, which represents the occupation time of the process $X$ below $-y$ until an independent exponential time.
 In the following, denote $A=\frac{\phi^{+\prime}_{q+p}(-y)}{\phi^{+}_{q+p}(-y)}$, and note that $W_{q,1}(-y,m)+W_{q}(m, -y)A= B \phi^{+}_q(m)+C \phi^{-}_q(m)$ with $B=A \phi_q^{-}(-y)-\phi_q^{-\prime}(-y)$ and $C=\phi_q^{+\prime}(-y)-A \phi_q^{+}(-y)$. From \eqref{main2}, for all $p, q>0$, we have
\begin{align}
&E_{x}[e^{-p O^{q,0}_{y}}]=1-\exp\left( -\int_x^{\infty} \frac{W_{q,2}(-y,u)+W_{q,1}(u, -y)\frac{\phi^{+\prime}_{q+p}(-y)}{\phi^{+}_{q+p}(-y)}}{W_{q,1}(-y,u)+W_{q}(u, -y)\frac{\phi^{+\prime}_{q+p}(-y)}{\phi^{+}_{q+p}(-y)}}du    \right)\notag\\
&\quad - \int_x^{\infty} \exp\left( -\int_x^m \frac{W_{q,2}(-y,u)+W_{q,1}(u, -y)\frac{\phi^{+\prime}_{q+p}(-y)}{\phi^{+}_{q+p}(-y)}}{W_{q,1}(-y,u)+W_{q}(u, -y)\frac{\phi^{+\prime}_{q+p}(-y)}{\phi^{+}_{q+p}(-y)}} du       \right)\times \frac{ \frac{p}{q+p} s^{\prime}(m) \frac{\phi^{+\prime}_{q+p}(-y)}{\phi^{+}_{q+p}(-y)} }{W_{q,1}(-y,m)+W_{q}(m, -y)\frac{\phi^{+\prime}_{q+p}(-y)}{\phi^{+}_{q+p}(-y)}}dm\notag\\
&=1-\lim\limits_{u\rightarrow \infty}\frac{W_{q,1}(-y,x)+W_{q}(x, -y)A}{W_{q,1}(-y,u)+W_{q}(u, -y)A} \notag\\
&\quad \quad -\int_x^{\infty} \frac{W_{q,1}(-y,x)+W_{q}(x, -y)A}{W_{q,1}(-y,m)+W_{q}(m, -y)A}  \times \frac{\frac{p}{q+p}s^{\prime}(m)A  }{W_{q,1}(-y,m)+W_{q}(m, -y)A}  dm\notag\\
&=1+\frac{p}{q+p} \frac{A}{C}\phi_q^{+}(x)-(B \phi^{+}_q(x)+C \phi^{-}_q(x)) \lim\limits_{u\rightarrow \infty} \frac{1+\frac{p}{q+p} \frac{A}{C}\phi_q^{+}(u)}{B \phi^{+}_q(u)+C \phi^{-}_q(u)}   \notag\\
&= 1+\frac{p}{q+p} \frac{A}{C}\phi_q^{+}(x)-(B \phi^{+}_q(x)+C \phi^{-}_q(x)) \frac{p}{q+p} \frac{A}{BC}\notag\\
&= 1-\frac{p}{q+p} \frac{\phi_{q+p}^{+\prime}(-y)\phi_q^{-}(x)  }{\phi_{q+p}^{+\prime}(-y)\phi_q^{-}(-y)-\phi_q^{-\prime}(-y) \phi_{q+p}^{+}(-y)}, \label{check2}
\end{align}
where the second last equality is due to the fact $\lim\limits_{u\rightarrow \infty}\phi_q^{+}(u)=\infty$.
%Plug in the definition of $A,B$ and $C$, we have the final expression
%\begin{align}
%E_{x}[e^{-p O^{q,0}_y}]&= \label{check3}
%\end{align}
To be consistent with the notations in Li and Zhou \citeyear{LZ13}, define $\psi_q^{\pm}(\cdot)=\pm\frac{\phi_q^{\pm\prime}(\cdot)}{\phi_q^{\pm}(\cdot)}$. If we take  $x=0$ and $y=-x=0$, then \eqref{check2} becomes
\begin{align}
E_{0}[e^{-p O^{q,0}_{0}}]&= 1-\frac{p}{q+p} \frac{\phi_{q+p}^{+\prime}(0)\phi_q^{-}(0)  }{\phi_{q+p}^{+\prime}(0)\phi_q^{-}(0)-\phi_q^{-\prime}(0) \phi_{q+p}^{+}(0)}=\frac{\frac{q}{q+p}\psi_{q+p}^{+}(0) +\psi_q^{-}(0)}{\psi_{q+p}^{+}(0)+\psi_q^{-}(0)}, \label{check3}
\end{align}
which agrees with Theorem $3.1$ of Li and Zhou \citeyear{LZ13} with substitutions $p\rightarrow \lambda$ and $q\rightarrow \delta$ there. We further generalize the Corollary $3.2$ of Li and Zhou \citeyear{LZ13} to non-zero levels as follows.
\begin{proposition}
In general, for $x\geq b$
\begin{align}
E_x \left[ e^{-p\int_0^{e_q} \1_{\lbrace X_t <b  \rbrace}dt }  \right]&= 1-\frac{p}{q+p} \frac{\phi_{q+p}^{+\prime}(b)\phi_q^{-}(x)  }{\phi_{q+p}^{+\prime}(b)\phi_q^{-}(b)-\phi_q^{-\prime}(b) \phi_{q+p}^{+}(b)}.\label{18}
\end{align}
and for $x<b$
\begin{align}
E_x \left[ e^{-p\int_0^{e_q} \1_{\lbrace X_t <b  \rbrace}dt }  \right]&=\frac{p}{q+p} \frac{\phi_{q+p}^{+}(x)}{\phi_{q+p}^{+}(b)}\left(1- \frac{\phi_{q+p}^{+\prime}(b)\phi_q^{-}(b)  }{\phi_{q+p}^{+\prime}(b)\phi_q^{-}(b)-\phi_q^{-\prime}(b) \phi_{q+p}^{+}(b)}\right)+\frac{q}{q+p}.\label{19}
\end{align}
\end{proposition}
\begin{proof}
For $x\geq b$, take $y=-b\geq -x$ in \eqref{check2}, and we have the desired result in \eqref{18}.
%\begin{align}
%E_x \left[ e^{-p\int_0^{e_q} \1_{\lbrace X_t <b  \rbrace}dt }  \right]&=P_x(\tau_b^{-}<e_q) E_b \left[ e^{-p\int_0^{e_q} \1_{\lbrace X_t <b  \rbrace}dt }  \right]+P_x(e_q< \tau_b^{-}).
%\end{align}
%Then \eqref{18} follows from taking $y=-b$ and $x=b$ in equation \eqref{check2} and Lemma $2.2$ of Zhang \citeyear{Z13}.

For $x<b$, from the memoryless property of $e_q$ and the strong Markov property of $X$, we have
\begin{align}
E_x \left[ e^{-p\int_0^{e_q} \1_{\lbrace X_t <b  \rbrace}dt }  \right]&=E_x \left[e^{-p \tau_b^{+}}; \tau_b^{+}<e_q\right] E_b \left[ e^{-p\int_0^{e_q} \1_{\lbrace X_t <b  \rbrace}dt }  \right]+E_x \left[e^{-p e_q}; e_q<\tau_b^{+}\right]\notag\\
&=E_x \left[e^{-(q+p) \tau_b^{+}}\right]E_b \left[ e^{-p\int_0^{e_q} \1_{\lbrace X_t <b  \rbrace}dt }  \right] +\frac{q}{q+p} \left(1-E_x \left[e^{-(q+p) \tau_b^{+}}\right]\right).
\end{align}
Then \eqref{19} follows from taking $y=-b$ and $x=b$ in equation \eqref{check2} and Lemma $2.2$ of Zhang \citeyear{Z13}. \qed
\end{proof}
\begin{re}
If we take $b=0$ in \eqref{18} and \eqref{19}, then they reduce to equations (18) and (19) in Corollary $3.2$ of Li and Zhou \citeyear{LZ13}. Note that \eqref{18} and \eqref{19} are equal at the boundary case when $x=b$.
\end{re}

Now we propose the Omega risk model with tax, and assume that $U$ is the after-tax surplus value process given in \eqref{U}, which corresponds to the Azema-Yor process in \eqref{ay} with $g(u)=u-\overline{\gamma}(u)=\int_x^u \gamma(z)dz$. For $y\geq 0$ and the bankruptcy rate function $\omega(.)>0$, define the \textit{time of bankruptcy} as
\begin{align}
\hat{\tau}_{\omega}&:=\inf\left\{t\geq 0: \int_0^t \omega(U_s)\1_{\lbrace U_s<-y\rbrace}ds >e_1  \right\},\label{ttax}
\end{align}
where $e_1$ is an independent exponential random variable with unit rate.
Now we are in the position to give our main result on the Laplace transform of $\hat{\tau}_{\omega}$ in the Omega risk model with tax.

\begin{theorem}\label{omega1}   If the bankruptcy rate function $\omega(\cdot)=\omega$ for a positive constant $\omega$, then for $q>0$, $y\geq -x$, the Laplace transform of the time of bankruptcy in the Omega risk model with tax is
 \begin{align}
E_{x}\left[e^{-q \hat{\tau}_{\omega}}\right]&=1-E_x\left[  e^{-\omega \int_0^{e_q} \1_{\lbrace U_s<-y\rbrace}ds}\right]\notag\\
&=\exp\left( -\int_x^{\infty} \frac{W_{q,2}(u-\overline{\gamma}(u)-y,u)+W_{q,1}(u, u-\overline{\gamma}(u)-y)\frac{\phi^{+\prime}_{q+\omega}(u-\overline{\gamma}(u)-y)}{\phi^{+}_{q+\omega}(u-\overline{\gamma}(u)-y)}}{W_{q,1}(u-\overline{\gamma}(u)-y,u)+W_{q}(u, u-\overline{\gamma}(u)-y)\frac{\phi^{+\prime}_{q+\omega}(u-\overline{\gamma}(u)-y)}{\phi^{+}_{q+\omega}(u-\overline{\gamma}(u)-y)}}du    \right)\notag\\
&\quad + \int_x^{\infty} \exp\left( -\int_x^m \frac{W_{q,2}(u-\overline{\gamma}(u)-y,u)+W_{q,1}(u, u-\overline{\gamma}(u)-y)\frac{\phi^{+\prime}_{q+\omega}(u-\overline{\gamma}(u)-y)}{\phi^{+}_{q+\omega}(u-\overline{\gamma}(u)-y)}}{W_{q,1}(u-\overline{\gamma}(u)-y,u)+W_{q}(u, u-\overline{\gamma}(u)-y)\frac{\phi^{+\prime}_{q+\omega}(u-\overline{\gamma}(u)-y)}{\phi^{+}_{q+\omega}(u-\overline{\gamma}(u)-y)}} du       \right)\notag\\
&\quad\quad\quad \times \frac{ \frac{\omega}{q+\omega} s^{\prime}(m) \frac{\phi^{+\prime}_{q+\omega}(m-\overline{\gamma}(m)-y)}{\phi^{+}_{q+\omega}(m-\overline{\gamma}(m)-y)} }{W_{q,1}(m-\overline{\gamma}(m)-y,m)+W_{q}(m, m-\overline{\gamma}(m)-y)\frac{\phi^{+\prime}_{q+\omega}(m-\overline{\gamma}(m)-y)}{\phi^{+}_{q+\omega}(m-\overline{\gamma}(m)-y)}}dm.\label{mainomegatax}
\end{align}
\end{theorem}
\begin{proof}
The desired expression is a consequence of Theorem \ref{mp} by substituting $\omega\rightarrow p$ and $u-\overline{\gamma}(u) \rightarrow g(u)$. \qed
\end{proof}

%{\cred Need more care when dealing with the probability of default. Can put in some heuristic analysis there for the case of $y=0$. Seems to be a tough job though.}

%Take $h(u)=g(u)=u-\overline{\gamma}(u)$ and $a\rightarrow \infty$ in Theorem \ref{mt}, we have
%\begin{corollary}\label{omega2}
%When the bankruptcy function is $\omega(x)=c$ for a positive constant $c$, for $y\geq 0$, the \textit{probability of bankruptcy} in the Omega risk model with tax is
% \begin{align}
%P_{x}\left(\hat{\tau}_{\omega}<\infty \right)
%&=1-.\label{mainomegatax2}
%\end{align}
%\end{corollary}
%\begin{re}
%In general, we can not directly let $q\rightarrow 0$ in Theorem \ref{omega1} to obtain the above result because the proof of Theorem \ref{mp} depends on conditioning on $\overline{X}_{e_q}$, and it may have a degenerate density when $q\rightarrow 0$(Zhang \citeyear{Z14}).
%\end{re}
%

\subsection{Applications of main results\label{sc}}

In this section, we look at some interesting applications of Theorem \ref{mt} and \ref{mp}.
%Some results are new to the literature and are of independent interest since they characterize the occupation time of a general Azema-Yor process until the first hitting time of an (possibly different) Azema-Yor process.

\subsubsection{\textbf{Occupation time of $U$ below $-y$ until $X$ first hits $-a$}}
\textit{ }
\vspace{0.1cm}

If $h(\cdot)=0\leq g(\cdot)$, then $V_t=X_t$ and $\tau_{0,a}=\tau^{-}_{-a}$, $P$-a.s. Thus the Laplace transform of $G_y^{a,0, g}=\int_0^{\tau^{-}_{-a}} \1_{\lbrace U_t<-y\rbrace}dt$ is given by \eqref{main} with substitution $0\rightarrow h(\cdot)$. If we consider the Omega risk model with tax and set $g(u)=u-\overline{\gamma}(u)$, then $G_y^{a,0, g}$ represents the occupation time of the after-tax process $U$ below $-y$ until the before-tax process $X$ first hits $-a$ from above.

\subsubsection{\textbf{Occupation time of $U$ below $-y$ until $U$ first hits $-a$}}
\textit{ }
\vspace{0.1cm}

If $h(u)=g(u)$, then $V_t=U_t$ and $\tau_{h,a}=\inf\lbrace t>0: U_t \geq -a \rbrace=:\tau_{g,a}$, $P$-a.s. Thus the Laplace transform of $G_y^{a,g, g}=\int_0^{\tau_{g, a}} \1_{\lbrace U_t<-y\rbrace }dt$ is given by \eqref{main} with substitution $g(\cdot) \rightarrow h(\cdot)$. If we consider the Omega risk model with tax and set $g(u)=u-\overline{\gamma}(u)$, then $G_y^{a,g, g}$ represents the occupation time of the after-tax process $U$ below $-y$ until it first hits $-a$ from above.

Instead of using an independent exponential random variable $e_q$ as in Theorem \ref{omega1},  we define the time of bankruptcy as the first instant either when this occupation time exceeds a grace period or when the process hits $-a$. The value of $a$ is usually set large due to the following economical motivations: the firm is declared bankrupt if its surplus value goes below a less severe level $-y$ for a grace period(\textit{reorganization} stage of the  U.S. Chapter 11 Bankruptcy code), or that it is immediately bankrupt when its surplus value first goes below a very severe level $-a$(\textit{immediate liquidation} stage of the U.S. Chapter 7 Bankruptcy code). Similar considerations of using this occupation time to model bankruptcy have appeared in recent literature and for a pointer to the literature, please refer to Li \citeyear{L13} and the references therein. %{\cred cite Bin Li's thesis, quote chapter 7 and chapter 11...}

\subsubsection{\textbf{Occupation time of the relative drawdown of $X$ over size $\alpha$ until $X$ or $U$ first hits $-a$}}
\textit{ }
\vspace{0.1cm}

In market practice, drawdown events are often quoted in percentages rather than in the absolute sense. Assume $x>0$, and define the first relative drawdown of $X$ over a fixed size $\alpha\in(0,1)$ as
\begin{align}
\eta_{\alpha}&:=\inf\left\{ t>0: \frac{\overline{X}_t -X_t}{ \overline{X}_t} \geq \alpha \right\}=\inf\left\{ t>0: X_t-(1-\alpha)\overline{X}_t \leq 0 \right\},
\end{align}
and it has been studied for example in Hadjiliadis and Vecer \citeyear{HV06}, Pospisil, Vecer and Hadjiliadis \citeyear{PVH09}, Zhang and Hadjiliadis \citeyear{ZH10}, and it is of particular importance to the modeling of ``Market Crashes"(Zhang and Hadjiliadis \citeyear{ZH12}). For a pointer to the recent literature, please refer to Zhang \citeyear{Z10} and the references therein.

If $h(u)=(1-\beta)(u-x)$ for $\beta\in[\alpha, 1]$, $g(u)=(1-\alpha)(u-x)\geq h(u)$, $y=(1-\alpha)x$ and take $a>(1-\alpha)x$, then $\1_{\lbrace U_t <-y \rbrace}=\1_{\lbrace (\overline{X}_t -X_t)/(\overline{X}_t) > \alpha \rbrace}$.
%$\eta_{\alpha}=\inf\lbrace t>0:  X_t-(1-\alpha) \overline{X}_t \leq 0 \rbrace=\tau_{g, y}$.
%Thus the occupation time of $U$ below $-y$ until a relative drawdown of $X$ over size $\alpha$ is given by  $G_y^{0, (1-\alpha)u, g}=\int_0^{\tau_{(1-\alpha)u, 0}} \1_{\lbrace U_t<-y\rbrace }dt$.
 Thus the Laplace transform of $G_{(1-\alpha)x}^{a,h,g}=\int_0^{\tau_{h,a}} \1_{\lbrace (\overline{X}_t -X_t)/(\overline{X}_t)> \alpha \rbrace}dt$ is given by \eqref{main} with the corresponding substitutions. It represents the occupation time of the relative drawdown of $X$ over size $\alpha$ until the after-tax process $V$(with constant tax rate $\beta$) first hits $-a$ from above. If $\beta=1$, then $h(u)=0$, and $G_{(1-\alpha)x}^{a,0,(1-\alpha)u}$ is the occupation time of the relative drawdown of $X$ over size $\alpha$ until the before-tax process $X$ first hits $-a$ from above. If $\beta=\alpha$. then $h(\cdot)=g(\cdot)$, and $G_{(1-\alpha)x}^{a,(1-\alpha)u,(1-\alpha)u}$ is the occupation time of the relative drawdown of the before-tax process $X$ over size $\alpha$ until the after-tax process $U$(with constant tax rate $\alpha$) first hits $-a$ from above.

Similarly, if we are in the setting of a generalized risk model with random observations  introduced in Albrecher, Cheung and Thonhauser (\citeyear{ACT11} \citeyear{ACT13}), then the time of bankruptcy is linked to the occupation time until an independent exponential time $e_q$, $q>0$. If  $g(u)=(1-\alpha)(u-x)$, $y=(1-\alpha)x$, then the Laplace transform of $O_{y}^{q,(1-\alpha)(u-x)}=\int_0^{e_q} \1_{\lbrace U_t <-y \rbrace} dt=\int_0^{e_q} \1_{\lbrace (\overline{X}_t -X_t)/(\overline{X}_t) > \alpha \rbrace} dt$ is given by \eqref{main2} with the corresponding substitutions. It represents the occupation time of the relative drawdown of $X$ over size $\alpha$ until an independent exponential time. It has applications in pricing a digital call on the \textit{relative} drawdown process with size $\alpha$ using a double Laplace inversion similar as in Section 5.2 of Zhang \citeyear{Z13}, where he considers the (absolute) drawdown process.

\subsubsection{\textbf{Occupation time of the relative drawdown of $V$ over size $\alpha$ until $V$ first hits $-a$}}
\textit{ }
\vspace{0.1cm}

Assume $x> 0$, and define the first relative drawdown over size $\alpha\in(0,1)$ for the Azema-Yor process $V$ as
\begin{align}
\eta_{\alpha, h}&:=\inf\left\{ t>0: \frac{\overline{V}_t -V_t}{ \overline{V}_t} \geq \alpha \right\}=\inf\left\{ t>0: X_t -\left((1-\alpha)(\overline{X}_t -x)+\alpha h(\overline{X}_t)< -y\right) \geq \alpha \right\},
\end{align}
where $y=(1-\alpha)x\geq -x$. If we take $g(u)=(1-\alpha)(u-x)+\alpha h(u)$, then $h(u)\leq  g(u) \leq u-x$ because $0\leq h(u)\leq u-x$. For $a>(1-\alpha)x$, we have
\begin{align}
G_{y}^{a,h, g}&=\int_0^{\tau_{h,a}} \1_{\left\{  \frac{\overline{V}_t -V_t}{ \overline{V}_t} > \alpha  \right\}}dt=\int_0^{\tau_{h,a}} \1_{\left\{  X_t -((1-\alpha)(\overline{X}_t-x)+\alpha h(\overline{X}_t))< -y  \right\}}dt,\label{a1}
\end{align}
whose Laplace transform is given by \eqref{main} with substitutions $(1-\alpha)(u-x)+\alpha h(u)\rightarrow g(u)$ and $(1-\alpha)x\rightarrow y$ there. The above result holds for an Azema-Yor process with general $h(\cdot)$. If we consider the Omega risk model with tax and set $h(u)=u-\overline{\gamma}(u)$, then \eqref{a1} represents the occupation time of the ``relative drawdown" of the after-tax process $V$ over size $\alpha$ until $V$ first hits $-a$.

\subsubsection{\textbf{Occupation time of the drawdown of $X$ until the first relative drawdown of $X$ over size $\alpha$}}
\textit{ }
\vspace{0.1cm}

Assuming $x>0$, if we take $h(u)=(1-\alpha)(u-x)$ and $a=(1-\alpha)x$ for $\alpha\in(0,1)$, then $\tau_{h,a}=\eta_{\alpha}$. If we take $g(u)=u-x \geq  h(u)$, then for $-x\leq y<(1-\alpha)x$, the Laplace transform of $G_{y}^{(1-\alpha)x, (1-\alpha)u, u-x}=\int_0^{\eta_{\alpha}}\1_{\lbrace \overline{X}_t-X_t > y\rbrace}dt$ is given by \eqref{main} with corresponding substitutions. It represents the occupation time of the drawdown process of $X$ above $y$ until the first relative drawdown of $X$ over size $\alpha$. The notation $C_{y}^a:=\int_0^{\sigma_a}\1_{\lbrace \overline{X}_t-X_t >y \rbrace}$ in equation (21) of Zhang \citeyear{Z13} measures the amount of time for the (absolute) drawdown process to finish the ``last trip" from $y$ to $a$. Our $G_{y}^{(1-\alpha)x, (1-\alpha)u, u-x}$ measures the occupation time of an absolute drawdown of more than $y$ until the first \textit{relative} drawdown over size $\alpha\in(0,1)$. This provides an alternative risk functional to measure both the absolute and relative drawdown risks.

\subsubsection{\textbf{Occupation time of the drawdown of $V$ until the first relative drawdown of $V$ over size $\alpha$}}
\textit{ }
\vspace{0.1cm}

Assuming $x>0$ and consider the Azema-Yor process $V$ with functional $h(\cdot)$. Since $\overline{V}_t=\overline{X}_t-h(\overline{X}_t)$, $P$-a.s., we have $\overline{V}_t-V_t=\overline{X}_t-X_t$, $P$-a.s., and that $V$ and $X$ have the same drawdown process.  If we take $\tilde{h}(u)=(1-\alpha)(u-x)+\alpha h(u)$ and $a=(1-\alpha)x$ for $\alpha\in(0,1)$, then $\tau_{\widetilde{h},a}=\eta_{\alpha, h}$. If we take $g(u)=u-x \geq  \tilde{h}(u)$, then for $-x\leq y<(1-\alpha)x$, the Laplace transform of $G_{y}^{(1-\alpha)x, \tilde{h}, u-x}=\int_0^{\eta_{\alpha, h}}\1_{\lbrace \overline{X}_t-X_t > y\rbrace}dt=\int_0^{\eta_{\alpha, h}}\1_{\lbrace \overline{V}_t-V_t > y\rbrace}dt$ is given by \eqref{main} with corresponding substitutions. If we take $h(u)=u-\overline{\gamma}(u)$, then it represents the occupation time of the drawdown process of the after-tax process $V$ above $y$ until the first relative drawdown of $V$ over size $\alpha$.

\subsection{Extending to integral functionals through time change}
Using the results in Theorem $1$ of Cui \citeyear{C13b} (or Theorem $3.2.1$ in the Ph.D. thesis Cui \citeyear{C13}), we are able to extend the Theorem \ref{mt} here to a more general integral functional.
%Previous literature considers the special case when $\omega(\cdot)$ is a constant or a piecewise constant(Albrecher, Gerber and Shiu \citeyear{AGS11}, Gerber, Shiu and Yang \citeyear{GSY12}, and Li and Zhou \citeyear{LZ13}). The motivation behind the extension is that the bankruptcy rate function should be a decreasing function(p.262 of Gerber, Shiu and Yang \citeyear{GSY12}), because economically the possibility of going bankrupt decreases if the surplus value increases.
The method is based on stochastic time change and the key steps are listed below. If we define a Boreal measurable function $b(x)>0$ and $\varphi_t=\int_0^t b^2(X_s)ds, t\geq 0$ to be consistent in notations, and assume some technical assumptions(Engelbert-Schmidt conditions), then from Theorem 1(i) of Cui \citeyear{C13b}, we have the following stochastic representation
\begin{align}
X_{t} &=S_{\int_{0}^{t} b^2(X_{s})ds}=S_{\varphi_t}, \quad P\text{-a.s.},\label{tc1}
\end{align}
and the process $S$ is a time-homogeneous diffusion satisfying the following SDE
\begin{align}
dS_{t} &=\frac{\mu(S_t)}{b^2(S_t)}dt +\frac{\sigma(S_t)}{b(S_t)}dB_t, \quad S_0=X_0=x.\label{xsde2}
\end{align}

 Let $\tau$ denote a $\mathcal{F}_t$-stopping time of $S_t$,  from Theorem 1(iii) of Cui \citeyear{C13b}, we have that $\varphi_{\tau}:=\int_0^{\tau} b^2(X_s)ds$ is a $\mathcal{G}_t$-stopping time and $\tau^S=\varphi_{\tau}$, $P$-a.s., where $\tau^S$ is the corresponding stopping time for $S_t$, and $\mathcal{G}_t=\mathcal{F}_{\varphi_t}$.

We have $\overline{X}_t:=\max\limits_{0\leq u\leq t}X_u= \max\limits_{0\leq u\leq t}S_{\varphi_u}= \max\limits_{0\leq u\leq \varphi_t} S_u=:\overline{S}_{\varphi_t}$, $P$-a.s.,  with the second equality due to \eqref{tc1} and the third equality due to continuity. Similarly as in \eqref{ay}, if we define $V^{\ast}_t =S_{t}-h(\overline{S}_{t})$ and $U^{\ast}_t =S_{t}-g(\overline{S}_{t})$, then $V_t=V^{\ast}_{\varphi_t}$, $P$-a.s., and $U_t=U^{\ast}_{\varphi_t}$, $P$-a.s. If we define a stopping time $\tau^{\ast}_{h, a}:=\inf\left\lbrace t> 0: V^{\ast}_t\leq -a     \right\rbrace$, then from Theorem 1(iii) of Cui \citeyear{C13b}, we have $\tau^{\ast}_{h, a}=\varphi_{\tau_{h, a}}$, $P$-a.s. Define the following integral functional $G_y^{a,h, g,b}:=\int_0^{\tau_{h, a}} b^2(X_t) \1_{\lbrace U_t<-y \rbrace}dt$. Apply the change of variables formula(Problem $3.4.5$ (vi), p.174 of Karatzas and Shreve \citeyear{KS91})
\begin{align}
G_y^{a,h, g,b}&:=\int_0^{\tau_{h, a}} b^2(X_t)\1_{\lbrace U_t<-y \rbrace}dt=\int_0^{\tau_{h, a}} b^2(X_t)\1_{\lbrace U^{\ast}_{\varphi_t}<-y \rbrace}dt\notag\\
&=\int_0^{\varphi_{\tau_{h, a}}} \1_{\lbrace U^{\ast}_t<-y \rbrace}dt\notag\\
&=\int_0^{\tau^{\ast}_{h, a}} \1_{\lbrace U^{\ast}_t<-y \rbrace}dt\notag\\
&=:G_y^{a,h, g, \ast},\label{kb}
\end{align}
where $G_y^{a,h, g, \ast}$ is the occupation time of $S$.
Thus we have translated the study of the integral functional $G_y^{a,h, g,b}$ to that of the occupation time of $U^{\ast}_t$ below $-y$ until $V^{\ast}_t$ first hits $-a$. Observe that $S_t$ is also a time-homogeneous diffusion with SDE \eqref{xsde2}, and we can apply Theorem \ref{mt} to $S_t$.

Define $\phi^{+, \ast}_q(\cdot)$ and $\phi^{-, \ast}_q(\cdot)$ respectively as the increasing and decreasing positive solutions of the \textit{Sturm-Liouville} ordinary differential equation for $S$: $\frac{1}{2}\sigma^2(x) f^{\prime\prime}(x)+\mu(x)f^{\prime}(x)=q b^2(x)f(x)$. Note that $S$ has the same scale function $s(\cdot)$ as $X$ (since $(\frac{\mu(\cdot)}{b^2(\cdot)})/(\frac{\sigma(\cdot)}{b(\cdot)})^2=\frac{\mu(\cdot)}{\sigma^2(\cdot)}$), there exists a positive constant $w^{\ast}_q$ such that $w^{\ast}_q s^{\prime}(x)=(\phi^{+, \ast}_q)^{\prime}(x) \phi^{-, \ast}_q(x)-(\phi^{+, \ast}_q)^{\prime}(x)\phi^{+, \ast}_q(x)$. Define the auxiliary functions $W^{\ast}_{q}(x,y):=\frac{1}{w^{\ast}_q}(\phi^{+, \ast}_q(x) \phi^{-, \ast}_q(y)-\phi^{+, \ast}_q(y)\phi^{-, \ast}_q(x))$, $W^{\ast}_{q,1}(x,y):=\frac{\partial}{\partial x}W^{\ast}_q(x,y)$ and $W^{\ast}_{q,2}(x,y):=\frac{\partial}{\partial y}W^{\ast}_{q,1}(x,y)$. Now we are in the position to provide the following general Laplace transform of $G_y^{a,h, g,b}$.

\begin{theorem}\label{mt2}

For $q\geq 0$, $-x\leq y<a$, if $g(u)\geq h(u)$ for $u\in[x,\infty)$, then we have
\begin{align}
E_x[ e^{-q G_y^{a,h, g, b}}; \tau_{h, a}<\infty ]&= E_x[ e^{-q \int_0^{\tau_{h, a}} b^2(X_t)\1_{\lbrace U_t<-y \rbrace}dt}; \tau_{h, a}<\infty ]\notag\\
 &=\int_x^{\infty} \frac{ \frac{s^{\prime}(m)}{W^{\ast}_q(g(m)-y, h(m)-a)}    }{1+\frac{s(m)-s(g(m)-y)}{s^{\prime}(g(m)-y)}\frac{W^{\ast}_{q,1}(g(m)-y, h(m)-a)}{W^{\ast}_q(g(m)-y, h(m)-a)}}\notag\\
&\quad \times \exp\left(-\int_x^m \frac{ \frac{s^{\prime}(u)}{s^{\prime}(g(u)-y)} \frac{W^{\ast}_{q,1}(g(u)-y, h(u)-a  )}{W^{\ast}_q(g(u)-y, h(u)-a  )}      }{1+ \frac{s(u)-s(g(u)-y)}{s^{\prime}(g(u)-y)}  \frac{W^{\ast}_{q,1}(g(u)-y, h(u)-a  )}{W^{\ast}_q(g(u)-y, h(u)-a  )}   } du   \right)    dm.\label{mainb}
\end{align}
\end{theorem}
\begin{proof}
The proof follows from \eqref{kb} with Theorem \ref{mt} applied to $S$ in \eqref{xsde2}. \qed
\end{proof}
\begin{re}
 If $b(\cdot)=1$, then \eqref{mainb} reduces to \eqref{main}. The integral functional $G_y^{a,h, g, b}$ represents the conditional ``stochastic area" swept by $X$ until the first hitting time of an Azema-Yor process $V$ to $-a$, with the condition being that another Azema-Yor process $U$ stays below $-y$. If we take $h(u)=0$ and $g(u)=u-x$, then for $y^{\prime}=y+x\geq 0$ and $a>y$, we have $G_{y}^{a,0, u-x, b}=\int_0^{\tau_{-a}^{-}} b^2(X_t) \1_{\lbrace \overline{X}_t -X_t >y^{\prime}\rbrace}dt$, and it represents the stochastic area(with a drawdown constraint) swept by $X$ until $X$ first hits $-a$ from above.
\end{re}

\section{Examples\label{s3}}
In this section we illustrate the main results using a Brownian motion with drift: $X_t=\sigma B_t+\mu t, X_0=x=0$, with state space $J=(-\infty, \infty)$, where $\mu\neq 0, \sigma>0$. From Section $6.1$ of Zhang \citeyear{Z13}, denote $\delta:=\frac{\mu}{\sigma^2}, \gamma:=\sqrt{\delta^2+\frac{2q}{\sigma^2}}$, then $s(x)=\frac{1}{\delta}(1-e^{-2\delta x})$, $\phi_q^{+}(x)=e^{(\gamma-\delta)x}$,  $\phi_q^{-}(x)=e^{-(\gamma+\delta)x}$, $w_q=\gamma$, $W_q(x,y)=2e^{-\delta(x+y)} \frac{\sinh[\gamma(x-y)]}{\gamma}$, and $\frac{W_{q,1}(x,y)}{W_q(x,y)}=\gamma \coth[\gamma(x-y)]-\delta$.

In the following, we consider the Omega risk model with tax having a constant bankruptcy rate $\omega(\cdot)=\omega>0$ and a constant tax rate $\gamma(\cdot)=c\in[0,1)$. If we take $h(u)=g(u)=c u$,  then from Theorem \ref{mt}, for $c\in[0,1)$, $q\geq 0$, and $0\leq y<a$
\begin{align}
&E_0[ e^{-q G_y^{a,h, g}}; \tau_{h, a}<\infty ]= \int_0^{\infty} \frac{ \frac{s^{\prime}(m)}{W_q(c m-y, c m-a)}    }{1+\frac{s(m)-s(c m-y)}{s^{\prime}(c m-y)}\frac{W_{q,1}(c m-y, c m-a)}{W_q(c m-y, c m-a)}}\notag\\
&\quad \quad\times \exp\left(-\int_0^m \frac{ \frac{s^{\prime}(u)}{s^{\prime}(c u-y)} \frac{W_{q,1}(c u-y, c u-a  )}{W_q(c u-y, c u-a  )}      }{1+ \frac{s(u)-s(c u-y)}{s^{\prime}(c u-y)}  \frac{W_{q,1}(c u-y, c u-a  )}{W_q(c u-y, c u-a  )}   } du   \right)    dm\notag\\
&\quad\quad= \frac{2\delta\gamma e^{-\delta(a-y)}}{\gamma\cosh[\gamma(a-y)]-\delta \sinh[\gamma(a-y)]} \int_0^{\infty}  \frac{e^{2\delta m}}{B e^{2\delta(1-c)m+2\delta y} -1} \left( \frac{B e^{2\delta y} -1}{B e^{2\delta(1-c)m+2\delta y} -1} \right)^{\frac{1}{1-c}}dm, \label{eq1g}
\end{align}
where
\begin{align}
B&=\frac{\gamma\coth[\gamma(a-y)]+\delta}{\gamma\coth[\gamma(a-y)]-\delta}.
\end{align}

From Theorem \ref{omega1}, for $c\in[0,1)$, $p, q>0$, and $y\geq 0$
\begin{align}
&E_{0}\left[e^{-q \hat{\tau}_{\omega}}\right]%&=1-E_x\left[  e^{-\omega \int_0^{e_q} \1_{\lbrace U_s<-y\rbrace}ds}\right]\notag\\
=\exp\left( -\int_0^{\infty} \frac{W_{q,2}(c u-y,u)+W_{q,1}(u,  c u-y)\frac{\phi^{+\prime}_{q+\omega}( c u-y)}{\phi^{+}_{q+\omega}( c u-y)}}{W_{q,1}( c u-y,u)+W_{q}(u,  c u-y)\frac{\phi^{+\prime}_{q+\omega}( c u-y)}{\phi^{+}_{q+\omega}( c u-y)}}du    \right)\notag\\
&\quad + \int_0^{\infty} \exp\left( -\int_0^m \frac{W_{q,2}( c u-y,u)+W_{q,1}(u,  c u-y)\frac{\phi^{+\prime}_{q+\omega}( c u-y)}{\phi^{+}_{q+\omega}( c u-y)}}{W_{q,1}( c u-y,u)+W_{q}(u,  c u-y)\frac{\phi^{+\prime}_{q+\omega}( c u-y)}{\phi^{+}_{q+\omega}( c u-y)}} du       \right)\notag\\
&\quad\quad\quad\quad\quad \times \frac{ \frac{\omega}{q+\omega} s^{\prime}(m) \frac{\phi^{+\prime}_{q+\omega}( c m-y)}{\phi^{+}_{q+\omega}( c m-y)} }{W_{q,1}( c m-y,m)+W_{q}(m,  c m-y)\frac{\phi^{+\prime}_{q+\omega}( c m-y)}{\phi^{+}_{q+\omega}( c m-y)}}dm\notag\\
\quad \quad &= \frac{\omega}{q+\omega}(\gamma^{\prime}-\delta)e^{-\delta y}\int_0^{\infty} \frac{e^{c\delta m} (\gamma \cosh[\gamma y]+\gamma^{\prime}\sinh[\gamma y])^{\frac{1}{1-c}}}{(\gamma \cosh[\gamma(1-c)m+\gamma y]+\gamma^{\prime}\sinh[\gamma(1-c)m+\gamma y])^{\frac{1}{1-c}+1}}  dm,  \label{eq2g}
\end{align}
where $\gamma^{\prime}=\sqrt{\delta^2+\frac{2(q+\omega)}{\sigma^2}}$.
\begin{re}
We only manage to simplify \eqref{eq1g} and \eqref{eq2g} in terms of a one-dimensional integral, and we shall use numerical integration to evaluate them in practice. If we take the limit $c\rightarrow 1$ in \eqref{eq1g}, then it reduces to the expression in Corollary $6.2$ of Zhang \citeyear{Z13}. If we take $c=0$ in \eqref{eq1g}, then the integral  can be explicitly evaluated and we have $E_0[ e^{-q G_y^{a,h, g}}; \tau_{h, a}<\infty ]=\frac{\gamma e^{-\delta (a+ y)}}{\gamma\cosh[\gamma(a-y)]+\delta \sinh[\gamma(a-y)]}$, which agrees with formula $(2.2.5.1)$ on p.298 of Borodin and Salminen \citeyear{BS02} with substitutions $1\rightarrow \sigma$ here and $0\rightarrow x$, $-y\rightarrow r$, $-a\rightarrow z$, $q\rightarrow \gamma$, $\gamma\rightarrow \sqrt{2p}$, $\delta\rightarrow \sqrt{2q}$ there.

If we take $c=0$ in \eqref{eq2g}, then the integral can be explicitly evaluated and we have  $E_{0}\left[e^{-q \hat{\tau}_{\omega}}\right]=1-E_x\left[  e^{-\omega \int_0^{e_q} \1_{\lbrace X_s<-y\rbrace}ds}\right]=\frac{\omega(\gamma\cosh[\gamma y]+\gamma^{\prime}\sinh[\gamma y])}{(q+\omega)\gamma}\frac{\gamma^{\prime}-\delta}{\gamma+\gamma^{\prime}}e^{-(\gamma+\delta)y}$, which agrees with a result in Sec. $5.1$ of Li and Zhou \citeyear{LZ13} with substitutions $0\rightarrow y$, $1\rightarrow \sigma$ here and $\omega\rightarrow \lambda$, $q\rightarrow \delta$, $\gamma^{\prime}-\delta\rightarrow \beta_{\delta+\lambda}^+$, $-(\gamma+\delta)\rightarrow \beta_{\delta}^{-}$ there. Note that it also agrees with formula $(2.1.4.1)$ on p.254 of Borodin and Salminen \citeyear{BS02}.
\end{re}

\section{Conclusion and future research\label{s4}}

We have explicitly characterized the Laplace transform of the occupation time of an Azema-Yor process below a constant level until the first passage time of another Azema-Yor process or until an independent exponential time. As an application, we have obtained the explicit Laplace transform of the time of bankruptcy in the ``Omega risk model with surplus-dependent tax" proposed in this paper. Future research will be in extending the results in this paper to the ``risk model with tax and capital injection" introduced in Albrecher and Ivanovs \citeyear{AI14}, where the surplus value process is both refracted at its running maximum and reflected at zero. It would also be interesting to apply results in Section \ref{sc} to designing new risk functionals taking into account both the absolute and relative drawdown risks.

\footnotesize

\bibliography{omegatax}

\end{document}